\documentclass{article}

\usepackage[no-test-for-array]{nicematrix}
\usepackage{mathrsfs}
\usepackage{graphicx}
\usepackage{amsmath,amsthm,mathabx}
\usepackage{algorithm,algorithmic}
\newtheorem{definition}{Definition}
\newtheorem{theorem}{Theorem}

\usepackage{appendix}

\usepackage{mathrsfs}

\usepackage{graphicx}

\usepackage{array}
\usepackage{arydshln}
\usepackage{multirow}
\usepackage{makecell}

\usepackage{geometry}
\begin{document}
\title{Deterministic Search on Complete Bipartite Graphs by Continuous Time Quantum Walk}

\author{Honghong Lin, Yun Shang\footnote{shangyun@amss.ac.cn}}
\maketitle
\begin{abstract}
This paper presents a deterministic search algorithm on complete bipartite graphs. Our algorithm adopts the simple form of alternating iterations of an oracle and a continuous-time quantum walk operator, which is a generalization of Grover's search algorithm. The success probabilities of previous works on spatial search on bicliques are strictly less than 1 or can only reach 1 asymptotically. We address the general case of multiple marked states, so there is a problem of estimating the number of marked states. To this end, we construct a quantum counting algorithm based on the spectrum structure of the search operator. This is a nontrivial example of quantum counting for spatial search. To implement the continuous time quantum walk operator, we perform Hamiltonian simulation in the quantum circuit model, which is achieved in constant time, that is, the complexity of the quantum circuit does not scale with the evolution time.
\end{abstract}



\section{Introduction}

Quantum walk is a fundamental and powerful tool for designing quantum algorithms that outperform the classical counterparts, such as searching~\cite{1,2}, element distinctness~\cite{3}, and graph isomorphism~\cite{4}. It has been demonstrated that they are universal for quantum computing~\cite{5,6} and can be implemented in a variety of physical systems~\cite{7,8,9}.

There are two types of quantum walks: discrete- and continuous-time. In this work we focus on the continuous-time quantum walk, which is a natural analog of the classical continuous-time random walk. Given a graph $G = (V,E)$ with $|V|=n$, let $p_j(t)$ be the probability associated with vertex $j$ at time $t$ and $P(t) = (p_1(t),\dots,p_n(t))^T$. A continuous-time random walk over $G$ is a stochastic Markovian process that evolves as per the master equation
\begin{equation}\label{master}
\frac{\rm d}{{\rm d}t}P(t) = LP(t),
\end{equation}
which can be seen as a discrete analog of the diffusion equation $\frac{1}{k}\frac{\partial u}{\partial t} = \Delta u$. Here $L$ is the discrete Laplacian of the graph (a discrete approximation of the Laplacian operator $\Delta$ in the continuum). The solution of this differential equation can be given in closed form as 
\begin{equation}
P(t) = {\rm e}^{Lt}P(0).
\end{equation}
Note that \eqref{master} closely resembles the Schr\"{o}dinger equation
\begin{equation}\label{Schr?dinger}
i\frac{\rm d}{{\rm d}t}|\psi\rangle = H|\psi\rangle,
\end{equation}
except that (\ref{master}) lacks the factor of $i$. A continuous-time quantum walk is then defined as the unitary process that evolves as per the Schr\"{o}dinger equation (\ref{Schr?dinger}) under any Hermitian Hamiltonian that respects the structure of $G$. We choose $H$ to be the adjacency matrix of the graph in this paper and the solution of equation \eqref{Schr?dinger} can be given in closed form as
\begin{equation}
|\psi(t)\rangle = {\rm e}^{-iAt}|\psi(0)\rangle.
\end{equation}

In this work, we employ continuous-time quantum walk to construct a deterministic spatial search algorithm. Search problem is one of the most important problems in computer science. Many of the computation tasks can be reduced to search problems or invoke search algorithms as a subroutine. One early milestone in quantum computation is Grover's celebrated algorithm for search in an unstructured database~\cite{10}. Since then, extensive research has been conducted to enhance and generalize Grover's algorithm. There are two avenues of investigation. One is to de-randomize the original algorithm, making the search exact~\cite{11,12}. Since error probability becomes negligible only when the database becomes sufficiently large, determinacy is vital in problems where the search space is relatively small. Moreover, de-randomization holds inherent value in itself. The other is to search on more general graphs, since Grover's unstructured search can be formulated as quantum walk search on a complete graph~\cite{13,14}.

We pursue both of these two directions here, searching on graphs and searching with certainty. Every graph $G$ of order $m+n$ can be regarded as an interdependent network with a block adjacency matrix
\begin{equation}
A = 
\begin{pmatrix}
A_1 & B \\
B^T & A_2
\end{pmatrix},
\end{equation}
where $A_1$ is the $m \times m$ adjacency matrix of the graph $G_1$ with $m$ vertices, $A_2$ is the $n \times n$ adjacency matrix of the graph $G_2$ with $n$ vertices, and $B$ is the $m \times n$ matrix interconnecting $G_1$ and $G_2$. Among the published papers thus far, graphs on which a deterministic search algorithm exist consist only of complete graphs (with a self-loop at each vertex)~\cite{14}, star graphs~\cite{15} and $2 \times n$ Rook graph ($K_2 \square K_n = K_2 \otimes I_n+I_2 \otimes K_n$)~\cite{16}. In this paper, we propose a deterministic algorithm for search on a complete bipartite graph with multiple marked vertices. Equations (\ref{complete}-\ref{complete_bipartite}) represent the adjacency matrices of these three types of graphs, which shows that they all have a certain completeness property (indicated by the all-1 submatrix $J_{m \times n}$). We believe that this completeness property is useful and crucial for the construction of a deterministic search algorithm on graphs, principally because it enables one to identify a small-dimensional invariant subspace in which calculation can possibly be done, and more importantly, the spectrum of the reduced search operator possesses sufficient symmetries.

Complete Graph: 
\begin{equation}\label{complete}
A_{K_{m+n}} = 
\begin{pmatrix}
J_{m \times m} & J_{m \times n} \\
J_{n \times m} & J_{n \times n}
\end{pmatrix}+
\begin{pmatrix}
I_m & \bf0\\
\bf0& I_n
\end{pmatrix}
\end{equation}

$2 \times n$ Rook Graph:
\begin{equation}\label{Rook}
A_{K_2 \square K_n} = 
\begin{pmatrix}
J_{n \times n} & \bf0 \\
\bf0 & J_{n \times n}
\end{pmatrix}+
\begin{pmatrix}
-I_n & I_n \\
I_n & -I_n
\end{pmatrix}
\end{equation}

Complete Bipartite Graph: 
\begin{equation}\label{complete_bipartite}
A_{K_{m,n}} = 
\begin{pmatrix}
\bf{0} & J_{m \times n} \\
J_{n \times m} & \bf{0}
\end{pmatrix}
\end{equation}

Search on complete bipartite graphs has been explored in different cases. Wong et al.~\cite{17} studied the distinction between Laplacian and adjacency matrix in continuous-time quantum walk search by analysing search on the complete partite graphs, while Rhodes \& Wong~\cite{18} investigated how discrete-time quantum walk search the complete bipartite graphs with different initial states. The calculation of the algorithmic parameter (runtime) of these two works varies according to the size of the two partite sets and the arrangement of the marked vertices. As shown in section \ref{search} (theorem \ref{search theorem}), our algorithm does not depend on a case by case analysis. Besides, the success probabilities of these algorithms are strictly less than 1 or only reach 1 asymptotically.

Qu et al. \cite{15} gave a deterministic search algorithm on star graphs with a single marked vertex, which is a special case of the complete bipartite graphs. In comparison, our search scheme is simpler, comprised of iterations of an oracle followed by a continuous time quantum walk operator. Moreover, we consider the general case of multiple marked vertices, therefore necessitating an estimation of the number of marked vertices. Approximate counting is one of the most fundamental problems in computer science. It has applications, for example, to the solution of $\bf NP$-complete problems, which may be phrased in terms of the existence of a solution to a search problem~\cite{19}. Recently, Bezerra et al.~\cite{20} presented a quantum counting algorithm on complete bipartite graphs based on discrete-time quantum walk. The quantum walk dynamics has an eight-dimensional invariant subspace and the eigenvalues of the reduced operator separate into two nonequivalent classes. Our search operator has a simpler spectrum structure, consisting of -1 and a conjugate pair. Since the conjugate eigenvalues encode the information of the number of marked vertices, we employ quantum phase estimation to estimate the number of marked vertices, which is a nontrivial generalization of the quantum counting algorithm for Grover's unstructured search.

In section \ref{search}, we present our main result, the deterministic search algorithm on complete bipartite graphs. In section \ref{counting}, we give the quantum counting algorithm on complete bipartite graphs based on our search operator. In section \ref{simulation}, we implement the search algorithm in quantum circuit model by simulating the walk operator. We conclude with general discussion of the problems considered here in section \ref{discussion}.

\section{Deterministic Spatial Search Algorithm}\label{search}

Grover's original search algorithm is the iteration of the operator $G = (I-2|\psi\rangle\langle\psi|)O$, where $O$ denotes the oracle and $|\psi\rangle$ is the uniform superposition of all states. One of the deterministic versions of Grover's search is to introduce an arbitrary phase to the reflection $(I-2|\psi\rangle\langle\psi|)$, obtaining a generalized search operator $G(\alpha) = (I-(1-{\rm e}^{i\alpha})|\psi\rangle\langle\psi|)O$~\cite{21}. By choosing an appropriate phase $\alpha$ and iteration number, one can obtain a final state that is the superposition of all marked vertices. Our search operator is a further generalization, replacing the Grover diffusion operator $(I-2|\psi\rangle\langle\psi|)$ with a quantum walk operator on complete bipartite graphs. By setting an appropriate walk time and iteration number, the system will reach a final state of superposition of all marked states.

For a complete bipartite graph $K_{m,n}$ with $m$, $n$ vertices in the two partite sets respectively, the Hilbert space in which the search takes place is span$\{ |0\rangle, |1\rangle, \dots, |m+n-1\rangle \}$. We assume the marked vertices only occur in, say, the order-$n$ part of the graph and denote the uniform superposition of the $k$ marked vertices as $|w\rangle = \frac{1}{\sqrt{k}}\sum\limits_{i=m}^{m+k-1} |i\rangle$. This assumption is equivalent to say that the oracle is local, acting on the subspace corresponding to the order-$n$ part:
\begin{eqnarray}\label{local_oracle}
O = 
\left\{
\begin{aligned}
& I, \,\text{in span}\{ |0\rangle,\dots,|m-1\rangle \}, \\
& 2|w \rangle\langle w|-I, \,\text{in span}\{ |m\rangle,\dots,|m+n-1\rangle \}.
\end{aligned}
\right.
\end{eqnarray}
(Note that the definition differs from the standard definition by a phase -1.) In fact, such an oracle can be simulated by a standard global oracle. The standard oracle $\widetilde{O}$ is defined by making use of an oracle qubit $|q\rangle$:
\begin{equation}
|i\rangle|q\rangle \overset{\widetilde{O}}{\to} |i\rangle|q \oplus f(i)\rangle,
\end{equation}
$f(i)$ equal 1 if vertex $i$ is marked and equal 0 otherwise. If the oracle qubit is initially in the state $|-\rangle = \frac{|0\rangle-|1\rangle}{\sqrt{2}}$, then the action of the oracle is
\begin{equation}
|i\rangle|-\rangle \overset{\widetilde{O}}{\to} (-1)^{f(x)}|i\rangle|-\rangle,
\end{equation}
leaving the oracle qubit unchanged. We append an ancillary qubit $|b\rangle$ to indicate which part a state belongs to: $|b=0\rangle|v\rangle$ means that $|v\rangle$ is in the order-$m$ part and $|b=1\rangle|v\rangle$ means that $|v\rangle$ is in the order-$n$ part. This is possible since the underlying graph is assumed to be known. Let $C_jZ$ be the Z gate controlled by the ancilla being in state $|b=j\rangle$ with target on the oracle qubit. The local oracles are implemented by
\begin{equation}
O_j = -C_jZ \cdot \widetilde{O} \cdot C_jZ
\end{equation}
with the oracle qubit initiated in $|+\rangle = \frac{|0\rangle+\rangle}{\sqrt{2}}$ state~\cite{21}. $O_1$ is then the oracle in (\ref{local_oracle}) and the corresponding circuit is shown in Figure~\ref{oracle}.
\begin{figure}[!t]
\centering
\includegraphics[width=7cm]{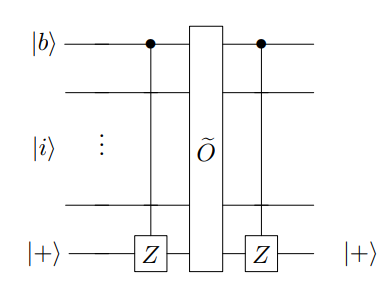}
\caption{Quantum circuit implementation for the oracle in (\ref{local_oracle}).}
\label{oracle}
\end{figure}
In the general case of marked vertices distributing in both partite sets, one can search in two steps, first for the marked vertices in one part (with the corresponding local oracle) and then the other part. Our quantum walk search operator is
\begin{equation}
U(t) = {\rm e}^{-iAt}O,
\end{equation}
where $A = A_{K_{m,n}}$ is the adjacency matrix of the complete bipartite graph.

Let $|s\rangle = \frac{1}{\sqrt{m}} \sum\limits_{i=0}^{m-1} |i\rangle$ be the uniform superposition of states corresponding to the order-$m$ part, and $| \overline{w} \rangle = \frac{1}{\sqrt{n-k}} \sum\limits_{i=m+k}^{m+n-1} |i\rangle$ be the complement of $|w\rangle$ in the order-$n$ part. These three states group together vertices that evolve identically by symmetry, thus $\text{span}\{ |s\rangle, |w\rangle, |\overline{w}\rangle \}$ is an invariant subspace of $A$, with the reduced adjacency matrix in this subspace being
\begin{equation}
A=
\begin{pmatrix}
0             & \sqrt{mk} & \sqrt{m(n-k)} \\
\sqrt{mk}     & 0         & 0             \\
\sqrt{m(n-k)} & 0         & 0             
\end{pmatrix}.
\end{equation}
This matrix has eigenvalues 0, $\pm\sqrt{mn}$ with eigenvectors
\begin{equation}
\begin{aligned}
& |v_{0}\rangle = \sqrt{\frac{n-k}{n}}|w\rangle - \sqrt{\frac{k}{n}}|\overline{w}\rangle, \\
& |v_{\pm\sqrt{mn}}\rangle = \frac{1}{\sqrt{2}}|s\rangle \pm \sqrt{\frac{k}{2n}}|w\rangle \pm \sqrt{\frac{n-k}{2n}}|\overline{w}\rangle,
\end{aligned}
\end{equation}
where the subscripts are the corresponding eigenvalues.

From this spectrum decomposition, we have
\begin{equation}
\begin{aligned}
{\rm e}^{-iAt}
& = |v_0 \rangle\langle v_0| + {\rm e}^{-it\sqrt{mn}}|v_{\sqrt{mn}} \rangle\langle v_{\sqrt{mn}}| + {\rm e}^{it\sqrt{mn}}|v_{-\sqrt{mn}} \rangle\langle v_{-\sqrt{mn}}| \\
& = 
\begin{pmatrix}
\cos(\sqrt{mn}t) & -i\sqrt{\frac{k}{n}}\sin(\sqrt{mn}t) & -i\sqrt{\frac{n-k}{n}}\sin(\sqrt{mn}t) \\
-i\sqrt{\frac{k}{n}}\sin(\sqrt{mn}t) & \frac{n-k}{n}+\frac{k}{n}\cos(\sqrt{mn}t) & -\frac{\sqrt{k(n-k)}}{n}(1-\cos(\sqrt{mn}t)) \\
-i\sqrt{\frac{n-k}{n}}\sin(\sqrt{mn}t) & -\frac{\sqrt{k(n-k)}}{n}(1-\cos(\sqrt{mn}t)) & \frac{k}{n}+\frac{n-k}{n}\cos(\sqrt{mn}t)
\end{pmatrix}.
\end{aligned}
\end{equation}
It follows that
\begin{equation}
U(t) = 
\begin{pmatrix}
\cos(\sqrt{mn}t) & -i\sqrt{\frac{k}{n}}\sin(\sqrt{mn}t) & i\sqrt{\frac{n-k}{n}}\sin(\sqrt{mn}t) \\
-i\sqrt{\frac{k}{n}}\sin(\sqrt{mn}t) & \frac{n-k}{n}+\frac{k}{n}\cos(\sqrt{mn}t) & \frac{\sqrt{k(n-k)}}{n}(1-\cos(\sqrt{mn}t)) \\
-i\sqrt{\frac{n-k}{n}}\sin(\sqrt{mn}t) & -\frac{\sqrt{k(n-k)}}{n}(1-\cos(\sqrt{mn}t)) & -\frac{k}{n}-\frac{n-k}{n}\cos(\sqrt{mn}t)
\end{pmatrix}.
\end{equation}

Solving the characteristic polynomial ${\rm det}(\lambda I - U) = (\lambda+1)(\lambda^2-[\frac{n-2k}{n}+\frac{2k}{n}\cos(\sqrt{mn}t)]\lambda+1)$, we obtain the eigenphases
\begin{equation}\label{P}
\pi, \,\theta_{\pm} = \pm 2\arcsin\left(\sqrt{\frac{k}{n}}\sin\big(\frac{\sqrt{mn}t}{2}\big)\right).
\end{equation}
The corresponding eigenvectors are
\begin{equation}\label{eigenvector}
\begin{aligned}
& |v_{-1}\rangle = \frac{1}{N} \Big(-i\sqrt{\frac{n-k}{n}}\sin\big(\frac{\sqrt{mn}t}{2}\big) |s\rangle + \cos\big(\frac{\sqrt{mn}t}{2}\big) |\overline{w}\rangle\Big), \\
& |v_{\pm}\rangle = \frac{1}{N} \Big(\frac{1}{\sqrt{2}}\cos\big(\frac{\sqrt{mn}t}{2}\big)|s\rangle \mp  \frac{1}{\sqrt{2}}\sqrt{1-\frac{k}{n}\sin^2\big(\frac{\sqrt{mn}t}{2}\big)}|w\rangle - i\frac{1}{\sqrt{2}}\sqrt{\frac{n-k}{n}}\sin\big(\frac{\sqrt{mn}t}{2}\big)|\overline{w}\rangle\Big),
\end{aligned}
\end{equation}
where $N = \frac{1}{\sqrt{1-\frac{k}{n}\sin^2(\frac{\sqrt{mn}t}{2})}}$ is the normalization factor. Define matrices $V = (|v_{-1}\rangle\, |v_+\rangle\, |v_-\rangle)$ and $D = {\rm diag}\{{\rm e}^{i\pi}, {\rm e}^{i\theta_+}, {\rm e}^{i\theta_-}\}$, then
\begin{equation*}
U^l = V D^l V^{\dagger}.
\end{equation*}

Since $|s\rangle$ is not the superposition of all states, we perform a preprocessing step, evolving $|s\rangle$ through the walk operator for time $\frac{t}{2}$ (the $\frac{1}{2}$ factor is chosen to facilitate the computation below):
\begin{equation}
\begin{aligned}
{\rm e}^{-iA\frac{t}{2}}|s\rangle = & \cos\big(\frac{\sqrt{mn}t}{2}\big)|s\rangle - i\sqrt{\frac{k}{n}}\sin\big(\frac{\sqrt{mn}t}{2}\big)|w\rangle - i\sqrt{\frac{n-k}{n}}\sin\big(\frac{\sqrt{mn}t}{2}\big)|\overline{w}\rangle.
\end{aligned}
\end{equation}
Using this state as the initial state, we evolve it through the unitary $U(t)$, then by an appropriate choice of the time $t$ and the iteration number $l$, we hope that the final state $U^l{\rm e}^{-iA\frac{t}{2}}|s\rangle$ have a unit overlap with $|w\rangle$.

We evaluate the value of $\langle w|U^l{\rm e}^{-iA\frac{t}{2}}|s\rangle$ by using the matrix representation of the states and operators in the subspace:
\begin{equation}
  \begin{aligned}
  \langle w|U^l{\rm e}^{-iA\frac{t}{2}}|s\rangle = &(0, 1, 0)V D^l V^{\dagger}{\rm e}^{-iA\frac{t}{2}}|s\rangle \\
  = &\frac{1}{N}\left(0, -\frac{1}{\sqrt{2}}, \frac{1}{\sqrt{2}}\right) D^l V^{\dagger}{\rm e}^{-iA\frac{t}{2}}|s\rangle \\
  = &\frac{1}{N}\left(0, -\frac{1}{\sqrt{2}}{\rm e}^{il\theta}, \frac{1}{\sqrt{2}}{\rm e}^{-il\theta}\right) V^{\dagger}{\rm e}^{-iA\frac{t}{2}}|s\rangle \\
  = &\left(-i\frac{1}{N}\cos(\frac{\sqrt{mn}t}{2})\sin(l\theta), \cos(l\theta), \frac{1}{N}\sqrt{\frac{n-k}{n}}\sin(\frac{\sqrt{mn}t}{2})\sin(l\theta)\right) {\rm e}^{-iA\frac{t}{2}}|s\rangle \\
  = &\sqrt{1-\frac{k}{n}\sin^2(\frac{\sqrt{mn}t}{2})}\sin(l\theta) + \sqrt{\frac{k}{n}}\sin(\frac{\sqrt{mn}t}{2})\cos(l\theta).
  \end{aligned}
\end{equation}

Here we make a variable substitution
\begin{equation}\label{substitution}
\sin(\frac{\sqrt{mn}t}{2}) = \sqrt{\frac{n}{k}}\sin(\frac{\pi}{2}x),
\end{equation}
it follows that $\sqrt{1-\frac{k}{n}\sin^2(\frac{\sqrt{mn}t}{2})} = \cos(\frac{\pi}{2}x)$ and
\begin{equation}
\theta = 2\arcsin(\sqrt{\frac{k}{n}}\sin(\frac{\sqrt{mn}t}{2})) = \pi x.
\end{equation}
The overlap now is
\begin{equation}
\begin{aligned}
  \langle w|U^l{\rm e}^{-iA\frac{t}{2}}|s\rangle = & \sin(l\theta)\cos(\frac{\pi}{2}x)+\cos(l\theta)\sin(\frac{\pi}{2}x) \\
  = &\sin\Big(l\theta+\frac{\pi}{2}x\Big) \\
  = & \sin\Big[(2l+1)\frac{\pi}{2}x\Big].
\end{aligned}
\end{equation}
Therefore, setting $x = \frac{1}{2l+1}$, or equivalently,
\begin{equation}\label{t}
\begin{aligned}
t = & \frac{2}{\sqrt{mn}}\arcsin\Big[\sqrt{\frac{n}{k}}\sin\big(\frac{\pi}{2}x\big)\Big] \\
  = & \frac{2}{\sqrt{mn}}\arcsin\Big[\sqrt{\frac{n}{k}}\sin\Big(\frac{\pi}{2(2l+1)}\Big)\Big],
\end{aligned}
\end{equation}
we have $\langle w|U^l{\rm e}^{-iA\frac{t}{2}}|s\rangle = 1$. However, $\sqrt{\frac{n}{k}}\sin(\frac{\pi}{2}x) \le 1$ from equation (\ref{substitution}), the requirement on $l$ is
\begin{equation}\label{l}
l \ge \frac{\pi}{4}\sqrt{\frac{n}{k}}-\frac{1}{2}.
\end{equation}

We formalize the above result into a theorem:
\begin{theorem}\label{search theorem}
Let $K_{m,n}$ be a complete bipartite graph with $k$ marked vertices on the order-$n$ part and $A$ be its adjacency matrix. For any $l \ge \frac{\pi}{4}\sqrt{\frac{n}{k}}-\frac{1}{2}$ and $t = \frac{2}{\sqrt{mn}}\arcsin\Big[\sqrt{\frac{n}{k}}\sin\Big(\frac{\pi}{2(2l+1)}\Big)\Big]$, starting from the state $|s\rangle = \frac{1}{m}\sum_{i=0}^{m-1}|i\rangle$, algorithm $U^l(t){\rm e}^{-iA\frac{t}{2}}$ returns a final state of uniform superposition of all marked vertices with certainty.
\end{theorem}

\section{Quantum Counting for Spatial Search}\label{counting}

The algorithmic parameters $l$ and $t$ depend on the number of marked vertices $k$, as indicated by equations (\ref{l}) and (\ref{t}). Therefore, in the general case of multiple marked vertices, it is necessary to estimate $k$ before applying the algorithm.

In the context of unstructured search, the initial state may be re-expressed as
\begin{equation}
|\psi\rangle = \sqrt{\frac{N-k}{N}}|\alpha\rangle + \sqrt{\frac{k}{N}}|\beta\rangle,
\end{equation}
where $|\alpha\rangle = \frac{1}{\sqrt{N-k}} \sum\limits_{i \text{\,unmarked}}|i\rangle$ and $|\beta\rangle = \frac{1}{\sqrt{k}} \sum\limits_{i \text{\,marked}}|i\rangle$. The Grover operator $G = (I-2|\psi\rangle\langle\psi|)O$ is a rotation operator of angle $\theta$ with eigenvalues ${\rm e}^{i\theta}$ and ${\rm e}^{i(2\pi-\theta)}$. Then it is possible to estimate the number of marked vertices with the phase estimation technique based on the observation that 
\begin{equation}\label{counting_Grover}
\sin^2\left(\frac{\theta}{2}\right) = \sin^2\left(\frac{2\pi-\theta}{2}\right) = \frac{k}{N}
\end{equation}

Apparently, this algorithm cannot be directly applied due to the different settings. Fortunately, the special spectrum structure of the algorithmic operator $U(t)$ allows us to construct a simple quantum counting algorithm by employing quantum phase estimation in a similar fashion. 

\begin{algorithm}
\floatname{algorithm}{Algorithm}
\renewcommand{\algorithmicrequire}{\textbf{Input:}}
\renewcommand{\algorithmicensure}{\textbf{Output:}}
\footnotesize
\caption{Quantum Phase Estimation (QPE)}
\label{QPE}
\begin{algorithmic}[1]
\REQUIRE a unitary $U$, initial state$|0\rangle^{\otimes p}|\phi_0\rangle$ ($p$ is the number of qubits in the first register and $|\phi_0\rangle$ is an arbitrary state of the second register).
\ENSURE $\tilde{\theta}$ (an estimate of one of the eigenphases of $U$).

\STATE Apply Hadamard gate $H$ to each qubit in the first register;
\STATE Apply $U^{2^j}$ on the second register controlled by qubit $j$, $1 \le j \le p$ in the first register;
\STATE Apply inverse quantum Fourier transform to the first register;
\STATE Measure the first register in the computational basis (assume the result is $l$);
\STATE \textbf{return} $\tilde{\theta} = 2\pi\frac{l}{M} (M=2^p)$.
\end{algorithmic}
\end{algorithm}

We adapt the notation in \cite{21} and give the following definition:
\begin{definition}
For any integer $M > 0$ and real number $0 \le \omega < 1$, let
\begin{equation*}
|S_M(\omega)\rangle = \frac{1}{\sqrt{M}}\sum_{y=0}^{M-1} {\rm e}^{2\pi i\omega y}|y\rangle.
\end{equation*}
We then have, for integers $0 \le j \le M-1$, the Fourier basis $F_M|j\rangle = |S_M(\frac{j}{M})\rangle$ for the space ${\rm span}\{|j\rangle, 0 \le j \le M-1\}$.
\end{definition}

We express $|\phi_0\rangle = \sum \alpha_j|\theta_j\rangle$ in the eigenbasis of $U$ and with probability $|\alpha_j|^2$, the state of the second register would be $|\theta_j\rangle$. After step 2 of the algorithm, the state of the first register will be
\begin{align}
\frac{1}{\sqrt{M}}\sum_{l=0}^{M-1} {\rm e}^{il\theta_j}|l\rangle & = |S_M(\frac{\theta_j}{2\pi})\rangle \\
& = \sum_{l=0}^{M-1} f_lF_M|l\rangle, \label{F}
\end{align}
provided that the state in the second register is $|\theta_j\rangle$ (with probability $|\alpha_j|^2$). We expand the state in Fourier basis in equation \eqref{F}. Step 3 performs the inverse Fourier transform, thus with probability $|f_l|^2$, step 4 will measure the result $l$. If $|\phi_0\rangle$ is one of the eigenvectors $|\theta_j\rangle$ of the unitary $U$, the following theorem from \cite{21} ensures that $\tilde{\theta} = 2\pi\frac{l}{M}$ is a good estimate of $\theta_j$ with high probability.
\begin{theorem}\label{probability}
If $\frac{\theta}{2\pi}M$ is an integer, then
\begin{equation}
{\rm Prob}\Big(l = \frac{\theta}{2\pi}M\Big) = 1.
\end{equation}
Otherwise,
\begin{equation}
\begin{aligned}
& {\rm Prob}\Big(l = \Big\lfloor \frac{\theta}{2\pi}M \Big\rfloor\Big) + {\rm Prob}\Big(l = \Big\lceil \frac{\theta}{2\pi}M \Big\rceil\Big) \\
= & \frac{\sin^2[M\Delta_1\pi]}{M^2\sin^2[\Delta_1\pi]} + \frac{\sin^2[M\Delta_2\pi]}{M^2\sin^2[\Delta_2\pi]} \ge \frac{8}{\pi^2}
\end{aligned}
\end{equation}
where $\Delta_1 = \frac{1}{M}(\frac{\theta_j}{2\pi}M-\lfloor \frac{\theta}{2\pi}M \rfloor)$, $\Delta_2 = \frac{1}{M}(\lceil \frac{\theta}{2\pi}M \rceil-\frac{\theta_j}{2\pi}M)$.
\end{theorem}

Intuitively, quantum phase estimation works as follows. The eigenvalues of a unitary lie on a unit circle. Divide the circle into $M$ parts. QPE can only obtain an estimation of the eigenphases from set $\Omega_M = \{\frac{j}{M}2\pi,0 \le j \le M-1\}$. Hence, for eigenphase $\theta$ in Figure~\ref{theta}, the best estimation is $\theta^-$ or $\theta^+$ as depicted in Figure~\ref{theta}, which correspond to the measure result of $\lfloor \frac{\theta}{2\pi}M \rfloor$ or $\lceil \frac{\theta}{2\pi}M \rceil$. Theorem~\ref{probability} then states that the probability of obtaining this best estimation is at least $\frac{8}{\pi^2}$. Specially, in the case of $\theta \in \Omega_M$, one obtain the result $\theta$ with certainty.

The problem concerning QPE is that the input state $|\phi_0\rangle$ is not an eigenvector in general and one would not be able to determine which eigenvalue corresponds to the estimation result. Fortunately, we can circumvent this difficulty as equation~\ref{counting_Grover} does in the context of unstructured search. The key is that the spectrum of our evolutionary operator consists of a conjugate pair and -1. We evaluate $k$ by rewriting equation \eqref{P} for $\theta_{\pm}$: 
\begin{equation}
k = \Bigg(\frac{\sin\frac{\theta_{\pm}}{2}}{\sin(\frac{\sqrt{mn}}{2}t)}\Bigg)^2n
\end{equation}
If we take $t = t_0 = \frac{\pi}{\sqrt{mn}}$, then $\theta_{\pm} = 2\arcsin\sqrt{\frac{k}{n}} \le 2\arcsin\frac{1}{\sqrt{2}} = \frac{\pi}{2}$ provided that $k < \frac{n}{2}$, and
\begin{equation}\label{k}
k = \Big(\sin\frac{\theta_{\pm}}{2}\Big)^2n.
\end{equation}
Eq.~(\ref{k}) is the counterpart of eq.~(\ref{counting_Grover}) in the context of spatial search. The eigenphase distribution of $U(t_0)$ on the unit circle is shown in Figure \ref{phase}. The eigenvectors in \eqref{eigenvector} now become
\begin{equation}
\begin{aligned}
& |v_{-1}\rangle = -i|s\rangle,\\
& |v_{\pm}\rangle = \frac{\mp|w\rangle-i|\bar{w}\rangle}{\sqrt{2}}.
\end{aligned}
\end{equation}
Note that the uniform superposition state
\begin{equation}
|\psi\rangle = \frac{1}{\sqrt{m+n}}\Big( i\sqrt{m}|v_{-1}\rangle + (i\sqrt{n-k}-\sqrt{k})\frac{|v_+\rangle-|v_-\rangle}{\sqrt{2}} \Big).
\end{equation}
Our quantum counting algorithm on the complete bipartite graph is to apply QPE to the operator $U(t_0)$ starting with the state $|0\rangle^{\otimes p}|\psi\rangle$. With probability $1-\frac{m}{m+n}$, $|\psi\rangle$ will be in either $|v_+\rangle$ or $|v_-\rangle$, Theorem \ref{probability} ensures that we can obtain a good estimate of the eigenphases $\theta_{\pm}$ with high probability, from which we can evaluate $k$.

\begin{figure}[!t]
\centering
\begin{minipage}[c]{0.48\textwidth}
\centering
\includegraphics[width=6cm]{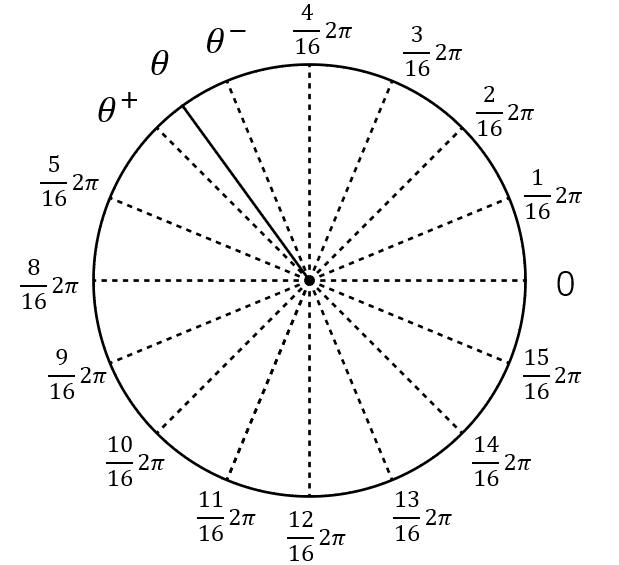}
\end{minipage}
\hspace{0.02\textwidth}
\begin{minipage}[c]{0.48\textwidth}
\centering
\includegraphics[width=6cm]{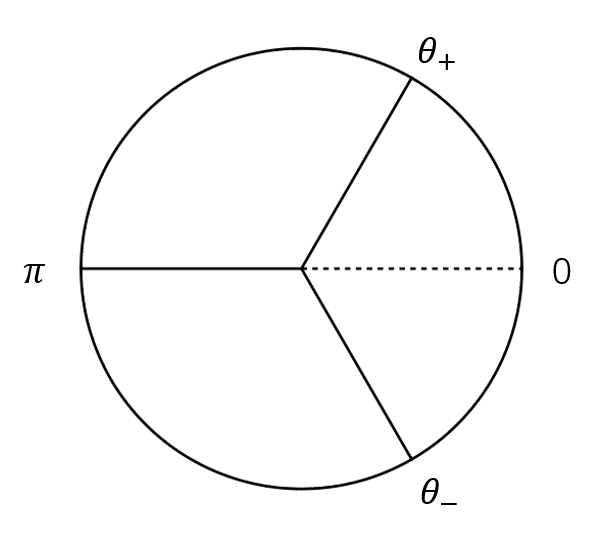}
\end{minipage}\\[3mm]
\begin{minipage}[t]{0.48\textwidth}
\centering
\caption{The possible phase set returned by Algorithm 1 for $M = 16$. The best estimation for an angle $\theta$ is either $\theta^+$ or $\theta^-$.}
\label{theta}
\end{minipage}
\hspace{0.02\textwidth}
\begin{minipage}[t]{0.48\textwidth}
\centering
\caption{The eigenphase distribution of $U(t_0)$ on the unit circle.}
\label{phase}
\end{minipage}
\end{figure}

\begin{theorem}[Quantum Counting on the Complete Bipartite Graph]
The error of the estimation $\tilde{k}$ of the number of marked vertices $k$ is upper bounded by
\begin{equation}
|\tilde{k}-k| \le \frac{2\pi}{M}\sqrt{k(n-k)}+\frac{\pi^2}{M^2}n
\end{equation}
with success probability at least $\frac{8}{\pi^2}$. If $\frac{\theta_{\pm}}{2\pi}M$ is an integer, then $\tilde{k}=k$ with certainty.
\end{theorem}
\begin{proof}
With probability $\frac{8}{\pi^2}$, $|\tilde{\theta}-\theta| \le \epsilon = \frac{2\pi}{M}$. Then $|\tilde{k}-k| = \Big|\Big(\sin\frac{\tilde{\theta}}{2}\Big)^2-\Big(\sin\frac{\theta}{2}\Big)^2\Big|n$ by \eqref{k}. Using trigonometric identities, we obtain
\begin{equation}
\begin{aligned}
\sin^2\Big( \frac{\theta}{2}+\frac{\epsilon}{2} \Big) - \sin^2\Big( \frac{\theta}{2} \Big) = & \cos\frac{\theta}{2}\sin\frac{\theta}{2}\sin\epsilon + \cos\theta\sin^2\frac{\epsilon}{2}, \\
\sin^2\Big( \frac{\theta}{2} \Big) - \sin^2\Big( \frac{\theta}{2}-\frac{\epsilon}{2} \Big) = & \cos\frac{\theta}{2}\sin\frac{\theta}{2}\sin\epsilon - \cos\theta\sin^2\frac{\epsilon}{2}.
\end{aligned}
\end{equation}
It follows that
\begin{equation}
|\tilde{k}-k| \le n\Big(\sqrt{1-\frac{k}{n}}\sqrt{\frac{k}{n}}\epsilon + \frac{\epsilon^2}{4}\Big) = \frac{2\pi}{M}\sqrt{k(n-k)}+\frac{\pi^2}{M^2}n.
\end{equation}
\end{proof}

\section{Circuit Implementation}\label{simulation}

The heart of circuit implementation of our algorithm lies at the implementation of operator ${\rm e}^{-iAt}$. We follow the method of diagonalization in \cite{22} to address the case where $m=2^{l_1}$ and $n=2^{l_2}$.

We first diagonalize $A = Q \Lambda Q^{\dagger}$ using its eigenbasis, with $\Lambda = \text{diag}( \sqrt{mn},0^{m-1},-\sqrt{mn},0^{n-1} )$ and
$$
Q = \begin{pNiceArray}{w{c}{1.25cm}ccc:w{c}{1cm}ccc}
\Block{4-4}<\Large>{H^{\otimes l_1}} &&& & &&& \\
&&& & &&& \\
&&& & &&& \\
&&& & &&& \\ 
\hdottedline
&&& & \Block{4-4}<\Large>{H^{\otimes l_2}} &&& \\
&&& & &&& \\
&&& & &&& \\
&&& & &&& \\
\end{pNiceArray}
\begin{pNiceArray}{cw{c}{0.75cm}cc:cw{c}{0.5cm}cc}
\frac{1}{\sqrt{2}} & && & \frac{1}{\sqrt{2}} &&& \\
                   & \Block{3-3}<\large>{I_{m-1}} && & &&& \\
                   & && & &&& \\
                   & && & &&& \\
\hdottedline
\frac{1}{\sqrt{2}} &&& & -\frac{1}{\sqrt{2}} & && \\
                   &&& &                     & \Block{3-3}<\large>{I_{n-1}} && \\
                   &&& &                     & && \\
                   &&& &                     & &&
\end{pNiceArray},
$$
where H is the Hadamard gate.

Here we expand the Hilbert space by $d$ dimensions so that $m = n+d$. Define $\widetilde{\Lambda} = \text{diag}( \sqrt{mn},0^{m-1},\\-\sqrt{mn},0^{m-1} )$ and
\begin{equation}
\widetilde{Q} = 
\begin{pNiceArray}{w{c}{1.25cm}ccc:w{c}{0.75cm}cc:w{c}{0.5cm}}
\Block{4-4}<\Large>{H^{\otimes l_1}} &&& & &&& \\
&&& & &&& \\
&&& & &&& \\
&&& & &&& \\ 
\hdottedline
&&& & \Block{3-3}<\Large>{H^{\otimes l_2}} && & \\
&&& & && & \\
&&& & && & \\
\hdottedline
&&& & && & I_d
\end{pNiceArray}
\begin{pNiceArray}{cw{c}{0.75cm}cc:cw{c}{0.5cm}cc}
\frac{1}{\sqrt{2}} & && & \frac{1}{\sqrt{2}} &&& \\
                   & \Block{3-3}<\large>{I_{m-1}} && & &&& \\
                   & && & &&& \\
                   & && & &&& \\
\hdottedline
\frac{1}{\sqrt{2}} &&& & -\frac{1}{\sqrt{2}} & && \\
                   &&& &                     & \Block{3-3}<\large>{I_{m-1}} && \\
                   &&& &                     & && \\
                   &&& &                     & &&
\end{pNiceArray}.
\end{equation}
Then for 
\begin{equation}
\widetilde{A} = \widetilde{Q}\widetilde{\Lambda}\widetilde{Q}^{\dagger},
\end{equation}
it is easily seen that
\begin{equation}
{\rm e}^{-i\widetilde{A}t} = \widetilde{Q}{\rm e}^{-i\widetilde{\Lambda}t}\widetilde{Q}^{\dagger}
\end{equation}
acts as ${\rm e}^{-iAt}$ in the original space and as identity in the expended subspace. Equations (\ref{Q_1}-\ref{e}) show how to implement these operators by elementary gates with $P_0=|0\rangle\langle0|$ and $P_1=|1\rangle\langle1|$ being the 2-dimensional projection operators.

\begin{equation}\label{Q_1}
\begin{aligned}
\begin{pNiceArray}{cw{c}{0.75cm}cc:cw{c}{0.75cm}cc}
\frac{1}{\sqrt{2}} & && & \frac{1}{\sqrt{2}} &&& \\
                   & \Block{3-3}<\large>{I_{m-1}} && & &&& \\
                   & && & &&& \\
                   & && & &&& \\
\hdottedline
\frac{1}{\sqrt{2}} &&& & -\frac{1}{\sqrt{2}} & && \\
                   &&& &                     & \Block{3-3}<\large>{I_{m-1}} && \\
                   &&& &                     & && \\
                   &&& &                     & &&
\end{pNiceArray} & = 
I_{2^{l_1+1}} + 
\begin{pNiceArray}{cw{c}{0.7cm}cc:cw{c}{0.7cm}cc}
\frac{1}{\sqrt{2}}-1 & && & \frac{1}{\sqrt{2}} &&& \\
                     & \Block{3-3}<\large>{\bf0} && & &&& \\
                     & && & &&& \\
                     & && & &&& \\
\hdottedline
\frac{1}{\sqrt{2}} &&& & -\frac{1}{\sqrt{2}}-1 & &&   \\
                   &&& &                       & \Block{3-3}<\large>{\bf0} && \\
                   &&& &                       & && \\
                   &&& &                       & &&
\end{pNiceArray} \\
& = I_{2^{l+1}} + (H-I_2) \otimes P_0^{\otimes l_1}
\end{aligned}
\end{equation}
\begin{equation}\label{Q_2}
\begin{aligned}
\begin{pNiceArray}{w{c}{0.75cm}ccc:w{c}{0.5cm}cc:w{c}{0.45cm}}
\Block{4-4}<\large>{H^{\otimes l_1}} &&& & && & \\
&&&                                  & && & \\
&&&                                  & && & \\
&&&                                  & && & \\
\hdottedline
&&& & \Block{3-3}<\large>{H^{\otimes l_2}} && & \\
&&& & && & \\
&&& & && & \\
\hdottedline
&&& & && & I_d
\end{pNiceArray} & = I_{2^{l_1+1}} + 
\begin{pNiceArray}{w{c}{0.75cm}ccc:w{c}{0.75cm}cc:w{c}{0.25cm}}
\Block{4-4}{H^{\otimes l_1}-I_{2^{l_1}}} &&& & &&& \\
&&&                                  & &&& \\
&&&                                  & &&& \\
&&&                                  & &&& \\
\hdottedline
&&& & && & \\
&&& & && & \\
&&& & && & \\
\hdottedline
&&& & && & 
\end{pNiceArray} + 
\begin{pNiceArray}{w{c}{0.5cm}ccc:w{c}{1cm}cc:w{c}{0.25cm}}
&&& & && & \\
&&& & && & \\
&&& & && & \\
&&& & && & \\
\hdottedline
&&& & \Block{3-3}{H^{\otimes l_2}-I_{2^{l_2}}} && & \\
&&& & && & \\
&&& & && & \\
\hdottedline
&&& & && &
\end{pNiceArray} \\
& = I_{2^{l_1+1}} + P_0 \otimes (H^{\otimes l_1}-I_{2^{l_1}}) + P_1 \otimes P_0^{\otimes (l_1-l_2)} \otimes (H^{\otimes l_2}-I_{2^{l_2}})
\end{aligned}
\end{equation}
\begin{equation}\label{e}
{\rm e}^{-i\widetilde{\Lambda} t} = 
\begin{pNiceArray}{cw{c}{0.75cm}cc:cw{c}{0.75cm}cc}
{\rm e}^{-i\sqrt{mn}t} & &&                              & &&& \\
                       & \Block{3-3}<\large>{I_{m-1}} && & &&& \\
                       & &&                              & &&& \\
                       & &&                              & &&& \\
\hdottedline
                   &&&   & {\rm e}^{i\sqrt{mn}t} & &&   \\
                   &&&   &                     & \Block{3-3}<\large>{I_{m-1}} && \\
                   &&&   &                     & && \\
                   &&&   &                     & &&
\end{pNiceArray} = 
\begin{pmatrix}
{\rm e}^{-i\sqrt{mn}t} & \\
                       & {\rm e}^{i\sqrt{mn}t}
\end{pmatrix} \otimes P_0^{l_1}
\end{equation}

The circuit for ${\rm e}^{-i\widetilde{A}t}$ is given in Figure~\ref{circuit}. As pointed out in~\cite{22}, the diagonalization approach confines the time-dependent of ${\rm e}^{-i\widetilde{A}t}$ to the diagonal matrix ${\rm e}^{-i\widetilde{\Lambda}t}$, which can be implemented by a controlled phase gate. This means that the walk operator can be simulated efficiently in constant time, that is, the complexity of the quantum circuit does not scale with the parameter $t$.

For $m,n$ that are not a power of 2, we can simply add vertices to the search space, none of which are marked, such that the sizes of the two partite sets are a power of 2. Alternatively, one can employ the oracle Hamiltonian simulation algorithm. Assuming access to a unitary oracle that implement the adjacency matrix, the implementation of the quantum walk operator can be achieved through sparse Hamiltonian simulation algorithm~\cite{23}.  

\begin{figure*}[ht]
  \centering
  \includegraphics[width=10cm]{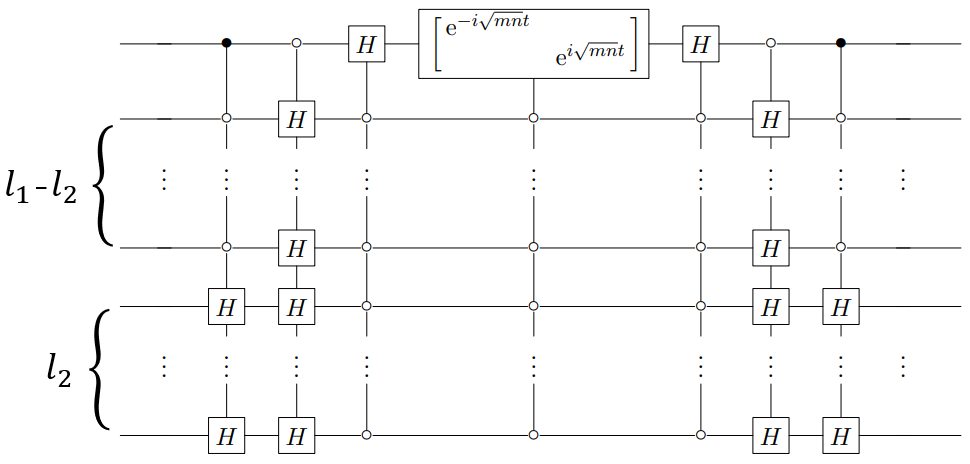}
  \caption{Quantum circuit for implementation of ${\rm e}^{-i\widetilde{A}t}$.}\label{circuit}
\end{figure*}

\section{Conclusion and Discussion}\label{discussion}

In this work we give a deterministic search algorithm on the complete bipartite graph. This is a new type of graphs that supports a deterministic search algorithm, other than complete graphs and $2 \times n$ Rook graph. Future work would be to construct deterministic search algorithms on other types of graphs. As we pointed out in Introduction and as the subsequent calculation shown, our algorithm strongly depends on the symmetries of the graph. Deterministic algorithm for more general graphs might require more sophisticated search scheme.

We also provide a simple algorithm for the quantum counting problem of spatial search based on our search operator. Quantum counting is another difficulty one might encounter in constructing deterministic search algorithms on graphs. Typically the algorithmic parameters such as iteration number or evolution time are functions of the number of marked vertices. However, the eigenphases of the evolution operator usually do not have sufficient symmetries, such that QPE cannot be used to extract the information of number of marked vertices from the operator eigenphases. In recent years, there are research works that have constructed quantum counting (and the closely related problem of amplitude estimation) algorithms without QPE \cite{24,25,26,27}. It is an interesting research line to construct quantum counting algorithms without QPE on graphs.






\end{document}